\documentclass[a4paper, 12pt]{article}
\usepackage{cmap} 
\usepackage[T2A]{fontenc}
\usepackage[utf8x]{inputenc}
\usepackage[english]{babel}
\usepackage{amsmath,amssymb}
\usepackage{amsthm}
\usepackage{graphicx,amssymb,amsmath,t1enc,hyperref}
\usepackage{yfonts}
\usepackage{hyperref}
\usepackage{subfig}
\usepackage{IEEEtrantools}
\usepackage{multirow}
\usepackage[]{algorithm2e}
\usepackage{dsfont}
\usepackage{tikz-cd}
\usepackage[left=3cm, right=3 cm]{geometry}
\usepackage{placeins}

\usetikzlibrary{arrows}
\tikzset
{
symbol/.style=
{
draw=none,
every to/.append style=
{
edge node={node [sloped, allow upside down, auto=false]{$#1$}}}
}
}

\theoremstyle{plain}
\newtheorem{theorem}{Theorem}
\newtheorem*{theorem*}{Theorem}
\newtheorem{definition}{Definition}
\newtheorem{lemma}{Lemma}

\newtheorem{proposition}{Proposition}

\newcommand{\N}{\ensuremath{\mathbb{N}}}
\newcommand{\R}{\ensuremath{\mathbb{R}}}
\newcommand{\Z}{\ensuremath{\mathbb{Z}}}
\newcommand{\T}{\ensuremath{\mathcal{T}}}

\newcommand{\A}{\ensuremath{\mathcal{A}}}
\newcommand{\HH}{\ensuremath{\mathcal{H}}}

\newcommand{\SSS}{\ensuremath{\mathcal{S}}}
\newcommand{\PPP}{\ensuremath{\mathcal{P}}}
\newcommand{\proj}{\ensuremath{\pi ^\prime}}

\DeclareRobustCommand*{\ora}{\overrightarrow}

\newcommand\numberthis{\addtocounter{equation}{1}\tag{\theequation}}

\title{Purely local growth of a quasicrystal}
\author{Thomas Fernique \and Ilya Galanov}

\begin{document}
\maketitle

\begin{abstract}
	Self-assembly is the process in which the components of a system, whether molecules, polymers, or macroscopic particles, are organized into ordered structures as a result of local interactions between the components themselves, without exterior guidance. In this paper, we speak about the self-assembly of aperiodic tilings. Aperiodic tilings serve as a mathematical model for quasicrystals - crystals that do not have any translational symmetry. Because of the specific atomic arrangement of these crystals, the question of how they grow remains open. In this paper, we state the theorem regarding purely local and deterministic growth of Golden-Octagonal tilings using the algorithm initially introduced in \cite{FG20}. Showing, contrary to the popular belief, that local growth of aperiodic tilings is possible.
\end{abstract}

\section{Introduction.}

Quasicrystals are physical solids with aperiodic atomic structure and symmetries forbidden in classical crystallography. They were discovered by Dan Shechtman in 1982 who subsequently won the Nobel prize for his discovery in 2011. 

It took Shechtman two years to publish his discovery, the paper \cite{SBGC84} with the report of analuminum-magnesium alloy with, as it was shown via electron diffraction, $10$-\emph{fold} or \emph{icosahedral} symmetry, \emph{i.e.}, rotational symmetry with angle $2\pi \over 10$ which is forbidden in periodic crystals by crystallographic restriction theorem. At the same time, sharp Bragg peaks of the diffraction pattern suggested long-range order in the new-found material. That was a clear violation of the fundamental principles of solid-state physics at the time. Two mentioned observations suggested that atoms in the material are structured in a non-periodic manner. The paper published in December 1984 coauthored by Levine and Steinhardt \cite{LeSt84} named the phenomenon as a \emph{quasicrystallinity} and the novel substance as a \emph{quasiperiodic crystal} or \emph{quasicrystal}.

To represent the peculiar atomic structure of quasicrystals Levine and Steinhardt in \cite{LeSt84} proposed the \emph{Penrose tilings} \cite{Pen74}. A tiling is a covering of a Euclidian plane by given geometric shapes without gaps and  overlaps. The set of basic shapes is called a \emph{prototile set} and the elements are called \emph{prototiles} or \emph{tiles}.

\begin{definition}
A prototile set is called aperiodic if admits only aperiodic tilings.
\end{definition}

Many of the properties of Penrose tilings can be derived from a so-called \emph{cut-and-project scheme}. This method, first introduced by DeBruijn \cite{DB81} in 1981, is based on the discovery that Penrose tilings can be obtained by projecting certain points from higher dimensional lattices to a $2$-dimensional plane. The method was subject of many generalizations, see \cite{BG13} for a comprehensive overview. Apart from Penrose tilings, it allows generating various other aperiodic tilings with symmetries forbidden in classical crystals but found in quasiperiodic ones.

Moreover, DeBruijn showed that the collection of $3$-patterns of a Penrose tiling uniquely determines the cutting slope in $5$-dimensional space. In general, the ability to fix the slope in higher dimensional space solely through finite patterns is referred to as \emph{local rules}. This subject was studied my many researchers \cite{Bee82,Lev88,Soc90,LPS92,Le95,Le97,Kat95,LP95,BF13,BF15a,BF15b,BF17,BF20}. In particular, a complete characterization of local rules for planar octagonal tilings is given in \cite{BF17}. Up to date aperiodic tilings and tilings generated via cut-and-project method is one of the main mathematical models to represent the atomical structure of a quasicrystal.

The discovery of Schechtman inspired other groups to search for quasicrystals. Within a few years, there had been reported quasicrystals with various other \emph{forbidden} symmetries, including decagonal \cite{Bend85}, pentagonal \cite{BaHe86}, octagonal \cite{WaCheKu87}, and dodecagonal \cite{Ish85}. Despite the abundance of quasicrystals synthesized in labs, there was not a consensus whether quasicrystallinity is a fundamental state of matter or the quasicrystals appear only as metastable phases under very specific (and unnatural) conditions. 

The common view was that aperiodic atomic structure is too complicated to be stable. Roger Penrose once said ``\emph{For this reason, I was somewhat doubtful that nature would actually produce such quasicrystalline structures spontaneously. I couldn't see how nature could do it because the assembly requires non-local knowledge}''. The discussion on what governs the stability of quasicrystals is still ongoing. Possible stability mechanisms include \emph{energetic stabilization} and \emph{entropy stabilization}, see~\cite{Boi06} for an overview.

Energy stabilization scenario suggests that quasicrystals can indeed be a state of minimal energy of the system, as in the case of classical crystals, and that short-range atomic interactions suffice to provide an aperiodically ordered structure. This is the case where the tilings with local rules are used as the main model.

Entropy stabilization scenario suggests that quasicrystals are always metastable phases and that aperiodic atomic order is governed by \emph{phason flips} (a local rearrangement of atoms that leave the free energy of the system unchanged) and structural disorder even if the state is not energetically preferred. This is the case of random tiling model \cite{Henley99}, \cite{Cohn00}.

Energy stabilization scenario would allow quasicrystals to be formed naturally. For instance, in~\cite{Nagao15} the growth of quasicrystals was directly observed with high-resolution transmission electron microscopy. Edagawa and his team produced a decagonal quasicrystal consisting of $\textrm{Al}_{70.8}\textrm{Ni}_{19.7}\textrm{Co}_{9.5}$. The growth process featured frequent errors-and-repair procedures and maintained nearly perfect quasiperiodic order at all times. Repairs, as concluded, were carried via phason flips, which is qualitatively different from the ideal growth models.

Another argument supporting the theory that quasicrystals can be stabilized via short-range interactions was given by Onoda \emph{et al} in~\cite{OSDS88} \cite{OSDS88,Soc91,HSS16}. They found an algorithm for growing a perfect Penrose tiling around a certain \emph{defective} seed using only the local information. A defective seed is a pattern made from Penrose rhombuses which is not a subset of any Penrose tiling. The finding broke the belief that non-local information is essential for building quasicrystalline types of structure.

Crystals can grow, both classic ones and the quasicrystals, local interactions between particles guide them into their respected places along the crystal lattice. Let us view the growth of crystals from the viewpoint of pure geometry. The question transforms into the following: how to algorithmically assemble the atomic structure of a growing crystal using only the local information? Classical crystals exhibit the unit cell which makes the question trivial: as long as we can see an instance of the unit cell, we know the structure. The question becomes much harder as we proceed to quasicrystals as there is no translational symmetry and no unit cell.

Here we aim to understand whether it is possible to grow an aperiodic tiling via a local self-assembly algorithm. The meaning of the locality constraint is the following: at each step, the algorithm must have access only to a finite neighborhood around a randomly chosen vertex of a seed - a finite pattern of a tiling we are trying to expand. Given the local neighborhood, the algorithm must identify the set of vertices which are to be added to the seed (or it may decide that there is not enough information to add a tile or a vertex and do nothing). Finally, the algorithm must not store any information between the steps.

The algorithm we propose, initially announced in \cite{FG20}, is a generalization of the OSDS rules algorithm. Similarly, we add the vertices if and only if they are uniquely determined by the neighbourhood but, and that is the crucial difference, we allow the forced vertices to be distanced from the growing pattern and not share any edges with it. This allows us to jump over the undefined areas and, as simulations suggest, and grow the \emph{cylinder set} of a given pattern or, as it was called in \cite{Gar89}, the \emph{empire}.

In this paper we prove the algorithm as able to grow particular aperiodic tilings with local rules namely the \emph{Golden-Octagonal tilings} initially defined in~\cite{BF15b}.

\begin{theorem*}
	For Golden-Octagonal tilings, for any finite pattern $\PPP$, there exists a finite seed $\SSS \supset \PPP$ and a growth radius $r$, such that the self-assembly Algorithm~\ref{algo:growth} builds the cylinder set (or empire) of $\PPP$.
\end{theorem*}

\section{Definitions and the main result.}
In this section we define and discuss properties of cut-and-project method, arguably the most versatile method to generate and to study aperiodic tilings. We define notion of tilings with local rules and introduce our self-assembly algorithm. In the end of the section, we formalize the mail result. 

Octagonal tilings are simply the tilings made of rhombuses with four distinct edge-directions. This gives us six rhombus prototiles in total: 

\begin{definition}
Consider \emph{the prototile set}: 
\[
	\{ \lambda \vec{v_i} + \mu \vec{v_j}, \quad 0 \le \lambda, \mu \le 1 \}, \qquad 0 \le i,j \le 3, \qquad i\ne j,
\]
where $\vec{v_i}$ and $\vec{v_j}$ are noncollinear vectors in $\R^2$. Any tiling of $\R^2$ with the above prototile set is called \emph{octagonal}.
\end{definition}

Our to-go method for generating octagonal tilings is canonical cut-and-project scheme. Here we are closely following chapter $7$ of \cite{BG13}.

Here we define and discuss properties of canonical cut-and-project method, our to-go method to generate octagonal tilings. Here we closely follow Chapter $7$ of \cite{BG13}.

\begin{definition}
	A $n \to d$ cut-and-project scheme consists of a physical space ${E\simeq\R^d}$, an  internal space $E^\perp \simeq \R^{n-d}$, a lattice $\Z^n$ in $E \times E^\perp = \R^n$ and the two natural projections $\pi : \R^n \to E$ and $\proj : \R^n \to E^\perp$ along with conditions that $\pi |_{\Z^n}$ is injective and that $\proj$ is dense in $E^\perp$.
\end{definition}

\begin{equation}
\begin{tikzcd}
\R^d \ar[d,symbol=\supset, ""']
	& \R^d \times \R^{n-d} \ar[l, "\pi"'] \ar[r, "\proj"] \ar[d,symbol=\supset]
	& \R^{n-d} \ar[d,symbol=\supset, "~dense"]  \\ 
	\pi(\Z^n) \arrow[rr, bend right, "*"']
		& \Z^n  \ar[l, "1-1"'] \ar[r, ""]
		&\proj (\Z^n)\\ 
\end{tikzcd}
\label{eq:cps}
\end{equation}

There is a bijection between $\pi(\Z^n)$ and $\proj(\Z^n)$, hence there is a well-defined map called the \emph{star map} between $E$ and $E^\perp$: 

\[
	x \to x^* := \proj\left((\pi_{|\Z^n})^{-1}(x)\right).
\]

A bounded subset $W$ of internal space with non-empty interior is called a \emph{window} or \emph{acceptance domain}. In the \emph{canonical} cut-and-project scheme, the window is chosen to be $W = \pi([0, 1]^n)$. The window acts as a filter for points which are to be projected. Using the star map, the set of vertices of a tilings or the \emph{projection set} can be written as

\[
	\Lambda(W):= \{ x \in \pi (\Z^n) \mid x^* \in W\}.
\]

Summary of a $n \to d$ canonical cut-and-project scheme is depicted in the diagram~\ref{eq:cps}.

\begin{definition}
A \emph{lift} is an injective mapping of vertices of a pattern to a subset of $\Z^n$ done in the following manner. Let $\{ v_i \}_{i=0}^{n}$ be set of edges of a rhombus tiling $\T$ up to a translation. First, we map each $v_i$ to a basis vector $e_i$ of $\Z^n$. Afterward, an arbitrary vertex is mapped to the origin. Then a vertex $x = \sum_{i = 1}^{n} a_i v_i$ of a pattern is lifted to $\sum_{i = 1}^{4} a_i e_i$. 
\end{definition}

Thus vertices of every octagonal tiling can be lifted to $\Z^4$. Since we are interested in tilings with local rules, first, we pick a subset of octagonal tilings whose lifted vertices are close to a plane, in other words the ones which can be seen as digitizations of surfaces in higher dimensional spaces:

\begin{definition}
An octagonal tiling is called \emph{planar} if there exists a 2-dimensional affine plane $E$ in $\R^4$ such that the tiling can be lifted into the stripe $E+[0,1]^4$. Then $E$ is called the \emph{slope} of the tiling. 
\end{definition}

Now, for planar octagonal tilings we define the notion of \emph{local rules} using an $r$-altas, the collection of all $r$-patterns, as follows: 

\begin{definition}
Vertices of a tiling which are at most $r$ edges away from a given vertex $x$ are called the \emph{$r$-pattern of $x$}.
\end{definition}

\begin{definition}
The set of all $r$-patterns of a tiling up to a translation is called the $r$-\emph{atlas} and denoted by $\A_{\T}(r)$. 
\end{definition}

\begin{definition}
	A planar tiling $\T$ is said to \emph{admit local rules} if there exists $r>0$ such that, any rhombus tiling $\T^\prime$ with $\A_{\T^{'}} (r) \subset \A_{\T}(r)$, $\T^{'}$ is also a planar tiling with the slope parallel to the slope of $\T$.
\end{definition}

\begin{definition}
Consider a patch $\PPP$ of a planar octagonal tiling with slope $E$ and window $W$.
The {\em cylinder} of $\PPP$, denoted by $[\PPP]$, is the intersection of all the planar octagonal tilings which contain $\pi\PPP$ as a subset and whose slope is a translation of $E$.
We denote by $W(\PPP)$ the intersection of all the translations of $W$ containing $\pi^\prime \PPP$.
\end{definition}
One has $W(\PPP)=W([\PPP])$. The vertices of $\pi^\prime [\PPP]$ are uniformly dense in $W([\PPP])$. Given a seed $\SSS$, nothing more than $[\SSS]$ can be grown by a deterministic algorithm no matter whether it is local or not.

\begin{algorithm}
\KwData{Growth radius $r$, $r$-atlas $\mathcal{A}$ and an seed pattern $\mathcal{S}=\SSS_0$}
\For{$k\geq 0$}
{
$\mathcal{S}_{k+1}\gets \mathcal{S}_k$\;
\For{$x$ vertex of $\mathcal{S}_k$}
{
$\mathcal{V}_x\gets$ vertices of $\mathcal{S}_k$ within distance $r$ of $x$\;
$\mathcal{A}_x\gets$ patterns in $\mathcal{A}$ translated by $x$ which contain $\mathcal{V}_x$\;
$\mathcal{F}_x\gets$ the vertices in common for all patterns in $\mathcal{A}_x$\;
$\mathcal{S}_{k+1}\gets \mathcal{S}_{k}\cup\mathcal{F}_x$\;
}
}
\caption{The self-assembly algorithm.}
\label{algo:growth}
\end{algorithm}

Now we are prepared to state the main result:
\begin{theorem}
	For Golden-Octagonal tilings, for any finite pattern $\PPP$ centered at the origin, there exists a seed $\SSS \supset \PPP$ and a growth radius $r$, such that at step $h$ of the self-assembly Algorithm~\ref{algo:growth}, it builds at least $[\PPP] \cap B_0(h)$.
\label{thm:main}
\end{theorem}

\newpage

\section{Preliminaries}

First, in this section, we define the notion of subperiod, one of the most useful tools which allow us to describe and characterize octagonal cut-and-project tilings with local rules. Then we state several lemmas in preparation to prove the main result in the next section.

\subsection{Subperiods} 

\begin{definition}
For a planar $4 \to 2$ tiling with a slope $E$, a \emph{subperiod of type} $i\in\{1,2,3,4\}$ is the smallest non-zero vector in $\R^4$ which belongs to $E$ and which has exactly $3$ integer coordinates along with an irrational $i$-th coordinate.
\end{definition}

In this paper our focus is mainly on a particular family of octagonal tilings: the Golden-Octagonal tilings are the planar $4 \to 2$ tilings generated via canonical cut-and-project method whose slope $E$ is generated by:
$$
u = \begin{pmatrix}
	-1 \\ 0 \\ \varphi \\ \varphi
\end{pmatrix}, \quad
v= \begin{pmatrix}
	0 \\ 1 \\ \varphi \\1 
\end{pmatrix},
$$
where $\varphi$ is the golden ratio. Its set of subperiods $\{q_i\}$ is written as: 

$$
q_1 = \begin{pmatrix}
	1-\varphi \\ 0 \\ 1 \\ 1
\end{pmatrix}, \quad
q_2 = \begin{pmatrix}
	1 \\ \varphi \\ 1 \\ 0 
\end{pmatrix}, \quad
q_3 = \begin{pmatrix}
	0 \\ 1 \\ \varphi \\ 1 
\end{pmatrix}, \quad
q_4 = \begin{pmatrix}
	1 \\ 1 \\ 0 \\ 1-\varphi 
\end{pmatrix}.
$$

Consider a subperiod $q_i$ of a planar tiling, we denote $\lfloor q_i \rfloor$ (respectively $\lceil q_i \rceil$) an integer vector which is equal to $q_i$ everywhere except for the irrational $i$-th coordinate, instead of which we take its floor (respectively ceiling). For example, for Golden-Octagonal tiling expressions for $\lfloor q_2 \rfloor$ and $\lceil q_2 \rceil$ are written as 

$$
\lfloor q_2 \rfloor = 
\begin{pmatrix}
	1 \\ \lfloor \varphi \rfloor \\ 1 \\ 0 
\end{pmatrix}
=
\begin{pmatrix}
	1 \\ 1 \\ 1 \\ 0 
\end{pmatrix},
\quad
\lceil q_2 \rceil = 
\begin{pmatrix}
	1 \\ \lceil \varphi \rceil \\ 1 \\ 0 
\end{pmatrix}
=
\begin{pmatrix}
	1 \\ 2 \\ 1 \\ 0 
\end{pmatrix}.
$$

For a point $A \in \Z^4$, $\proj A \in W$, it is obvious that either $\proj(A + \lfloor q_i \rfloor)$ or $\proj(A + \lceil q_i \rceil)$ is in $W$ as well. Iterating these sort-of-speech jumps along a subperiod gives us the following definition: 

\begin{definition}
For a planar octagonal tiling $\T$ with local rules with a window $W$ and vertex $X\in \T$, we define a \emph{subperiod line of vertex $X$ of type $i$} as the set:
\[
	Q_i(X) = \{y \in \Z^4~|~\pi^\prime y \in W, \quad  y_j=x_j+ q_j,\quad j \in \{0,1,2,3\} \setminus \{i\}  \}.
\]
\end{definition}

\begin{figure}[ht]
\centering
\begin{minipage}[b]{0.45\linewidth}
	\includegraphics[scale=0.6]{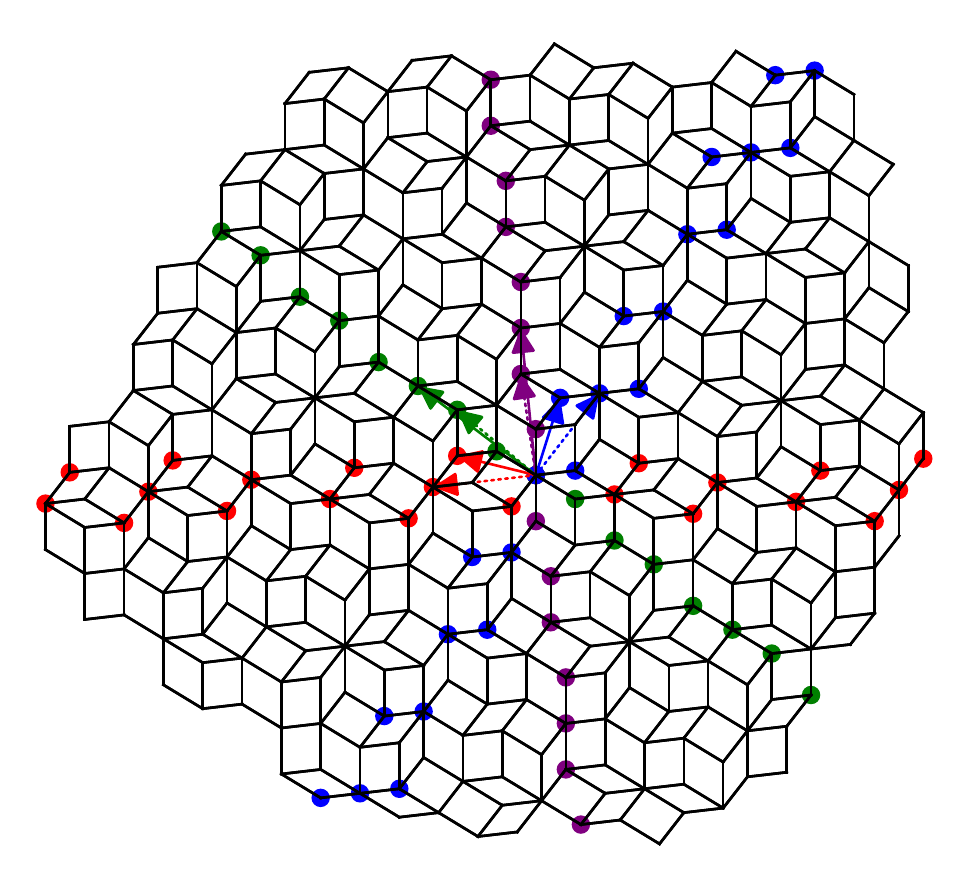}
  \label{fig:minipage1}
\end{minipage}
\quad
\begin{minipage}[b]{0.45\linewidth}
	\includegraphics[scale=0.6]{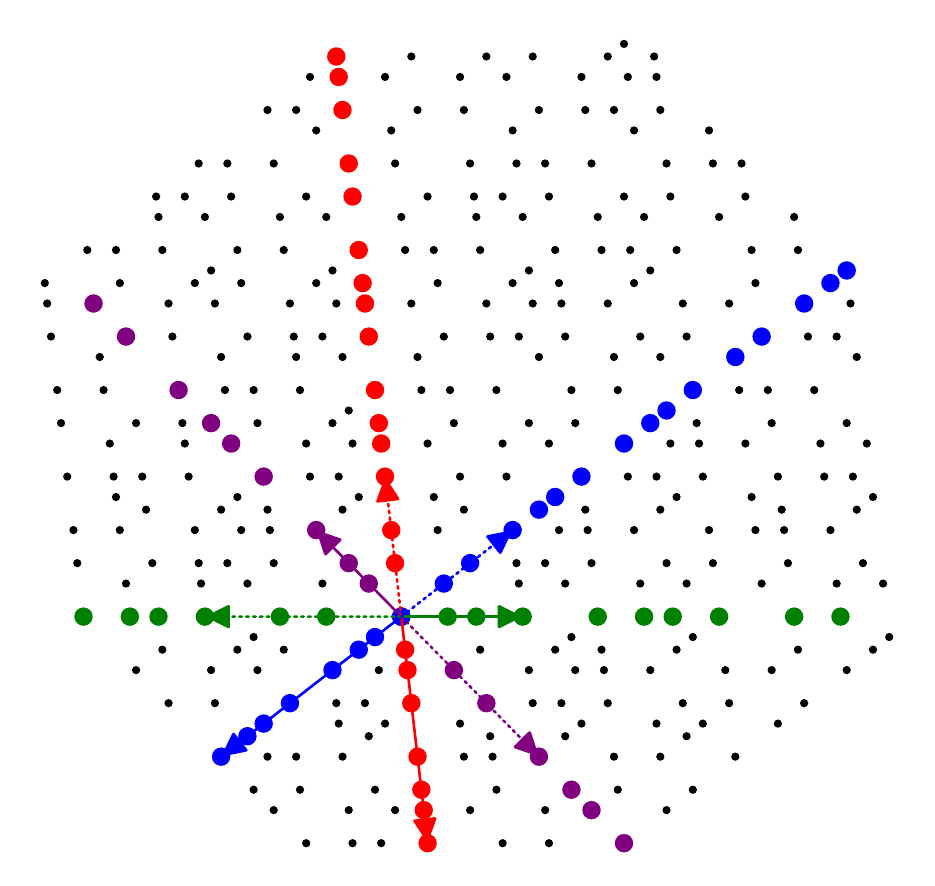}
\end{minipage}
  \caption{A pattern of Golden-Octagonal tiling is on the left and the same pattern projected to perpendicular space is on the right. Dotted arrows represent  $\lfloor q_i \rfloor$, solid arrows - $\lceil q_i \rceil$. Subperiods lines of a vertex in the center are depicted with different colors.}
  \label{fig:minipage2}
\end{figure}

\begin{definition}
	Consider a pattern $\PPP$, we say that a vertex $X \in \PPP$ \emph{sees} its $i$-th subperiod line if there exists $r>0$ such that: 
$$
	B_X(r) \cap \mathcal{P} \cap Q_i(X) \neq \emptyset.
$$
We say that a vertex $X$ sees $N$ vertices along its $i$-th subperiod if there exists $r>0$ such that:

$$
\mathbf{card} (B_X(r) \cap \mathcal{P} \cap Q_i(X)) \ge N.
$$
\end{definition}

A notable feature of the vertices of canonical cut-and-project tilings is that after the projection to the perpendicular space they are uniformly distributed in the window~\cite{Els85},\cite{Sch98}, \cite{BG13}, \cite{HKSW14}. In this paper, we use the weak version of mentioned result addressing only the distribution of vertices along subperiod lines:

\begin{lemma}[]
\label{lem:folk}
Consider a planar octagonal tiling $\T$ with at least one subperiod, a vertex $A \in \T$ and let $Q_i(A)$ be any of its subperiod lines. For any open line segment $(B,C) \subset W$ belonging to the line which contains $\proj Q_i(A)$, there exists $r>0$ inversely proportional to the length of the segment $(B,C)$, such that in $r$-neighborhood of $A$ there exists a vertex $X \in Q_i(A)$ such that $\pi^\prime X \in (B, C)$.
\end{lemma}

\subsection{Coincidences of subperiods} 

The following lemma highlights the interconnection between subperiods of different types in Golden-Octagonal tilings. We note that the lemma can be generalized to all the tilings with local rules. However, in this paper, we only prove it for Golden-Octagonal tilings:

\begin{lemma}[]
	\label{lem:third_intersection}
For Golden-Octagonal tiling, given $R\in \N$, there exists $l\in \N$, such that for any two integer open line segments $(A, A+e_i)$, $(B, B + e_j), i\neq j$, with the following properties:  their orthogonal projections to perpendicular space intersect at some point $X \in W$ and $\|A-B\|_1 < R$, there is another open line segment $(C, C+e_k)$, $k \neq i,j$, whose projection intersect both of the first two intervals in $X$ and $\|A-C\|_1 < lR$. Moreover, each of the three intervals has an endpoint that projects into $W$.
\end{lemma}

\begin{proof}
	Let $\{i,j,k\} = \{1,2,3\}$, without loss of generality let 
$$
A = (0,0,0,0), \quad B = (b_1,b_2,b_3,b_4), \quad b_1,b_2,b_3,b_4 \in \Z.
$$
By our assumption, there are points $X\in (A,A+e_1)$ and $Y\in(B,B+E_2)$ written as:

$$
X = (x, 0,0,0), \quad Y = (b_1, y, b_3, b_4), \quad x,y \in \R \setminus \Z,
$$
such that $\proj X = \proj Y = X$. Therefore, $Y - X \in E$ and:

\begin{align*}
	b_1 - x &=  - \lambda   
	\\
	y  &=  \mu 
	\\
	b_3 &=  \lambda \varphi + \mu \varphi 
	\\
	b_4 &=  \lambda \varphi + \mu 
\end{align*}
The last two equations yield:

\begin{align*}
	\lambda &= - b_3 + b_4 \varphi  \\
	\mu &= (b_3 - b_4) \varphi
\end{align*}
The first two equations then yield:
\begin{align*}
	x &= b_1 - b_3 + b_4 \varphi  \\
	y &= (b_3 - b_4) \varphi
\end{align*}
Therefore, $X$ and $Y$ can be written as
$$
X = (b_1 - b_3 + b_4 \varphi, 0, 0, 0), \quad Y = (b_1, (b_3 - b_4)\varphi, b_3, b_4).
$$
Now we search for the third intersection, let
$$
C = (c_1, c_2, c_3, c_4),\quad c_1,\dots,c_4 \in \Z,
$$
and
$$
 ]C, C+e_3[\ni Z = (c_1, c_2, z, c_4), \quad z\in\R\setminus\Z.
$$
Again, since we demand that $\proj Z = X$, we must have $Z - X \in E$: 
\begin{align*}
	c_1 - x &=  - \lambda   
	\\
	c_2  &=  \mu 
	\\
	z &=  \lambda \varphi + \mu \varphi 
	\\
	c_4 &=  \lambda \varphi + \mu 
\end{align*}
The second and fourth equations yield: 
\begin{align*}
	\lambda &=  (\varphi - 1)(c_4 -c_2)  \\
	\mu &= (b_3 - b_4) \varphi
\end{align*}
The first and third ones then yield:
\begin{align*}
	x &= c_1 +(\varphi -1)(c_4 - c_2)  \\
	z &= -c_2 + c_4 + c_2 \varphi
\end{align*}
Two expressions obtained for $x$ must be equal, that is: 
$$
c_4 - c_2 = b_4
$$
$$
c_1 - c_4 + c_2 = b_1 - b_3
$$
Consequently, the third intersection point $Z$ is written as: 
$$
Z = (b_1 - b_3 + b_4, c_2, c_2\varphi - b_4, b_4 + c_2),
$$
where $c_2$ can be chosen freely. Taking $c_2 < R$, since $\|A - B \|_1 < R$, it immediately follows that $\|C -A\|_1 < 8R$. 
Also, because the projection of each interval is a translation of an edge of $W$ and $X$ is in $W$, it follows that at least one endpoint of each interval also projects into $W$. Cases for other triples $\{i,j,k\}$ are treated in a similar way.
\end{proof}

\subsection{Shape of $r$-patterns} 

Consider an $n$-pattern $\PPP$ of a Golden-Octagonal tiling centered at a vertex $c \in \Z^4$:
\[
	\PPP =	\{ x\in \Z^4 \ \ | \ \ \| x-c \|_1 < n, \ \pi^\prime x \in W \}.
\]
To describe the geometrical shape of $\pi(\mathcal{P})$ we compute a section of a $3$-sphere of radius $r$ in $L_{1}$ distance with the generating plane of a Golden-Octagonal tiling. For example, one of the intersections of two of the $3$-dimensional faces of the sphere and the plane can be written as: 

\begin{align*}
\begin{cases}
x_1 + x_2 + x_3 + x_4 = r \\
- x_1 + x_2 + x_3 + x_4 = r \\
	(x_1, x_2, x_3, x_4)^T = \lambda u + \mu v
\end{cases}
\end{align*}
Solving the system we obtain 
\begin{align*}
	x = \frac{r}{2+\varphi} \begin{pmatrix}
           0 \\
           1 \\
           \varphi \\
           1
   \end{pmatrix},
\end{align*}
note that $x$, the coordinate of one of the vertices of the section, is proportional to $q_3$, the third subperiod. Calculating the other vertices of the section we obtain the following:

\begin{proposition}
\label{prop:shape}
There exists a set of constants $\{c_i \in \R^+\}$ for $ i = 1,\dots, 4$, such that for any $r\in\N$, vertices of an $r$-pattern of a Golden-Octagonal tiling centered in origin, when projected to the physical space, are within uniformly bounded distance from an octagon defined as a convex hull of  $~\pm c_i r \proj \vec q_i$, where $q_i$ is the $i$-th subperiod.
\end{proposition}

\subsection{Forcing vertices}

The last lemma in the section addresses the idea of how two vertices in a tilings can force the third one.

\begin{lemma}
\label{lem:cherez_odnu}
Consider a Golden-Octagonal tiling with a window $W$ and three points $B,A,C \in \Z^4$, such that there exists a translation of the slope such that $\proj B,\proj A$, and $\proj C$ belong to consecutive (clockwise or anticlockwise) edges of the translated window $W^\star$. Then, if $\proj B, \proj C \in W$, it follows that $\proj A \in W$.
\end{lemma}

\begin{proof}
	Trivial in the case of Golden-Octagonal tilings since the window is convex polygon without right angles, see Figure~\ref{fig:three_cons}.
\end{proof}

\begin{figure}[htb]
	\center{\includegraphics[width=0.75\textwidth]
	{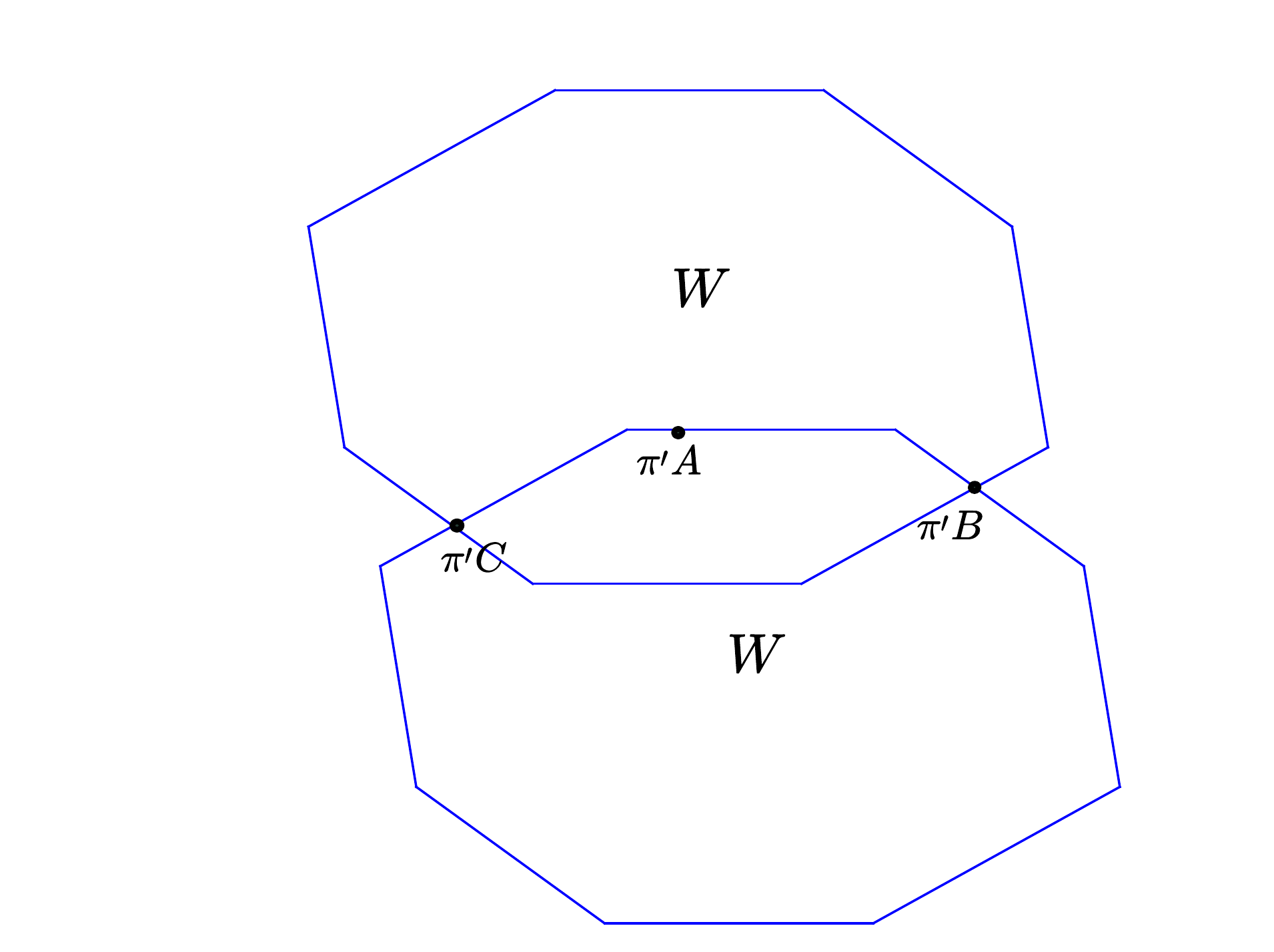}}
	\caption{\label{fig:three_cons} We have $\proj B, \proj A$, and $\proj C$ belonging to three consecutive edges of a translated window $W$. Note that there is no translation of the window $W$ such that $\proj B, \proj C \in W$ but $\proj A \notin W$.}
\end{figure}

\section{Proof of Theorem~\ref{thm:main}}

To state the main Lemma, we need to define so-called $\emph{extreme worms}$, the vertices of a pattern that are closest to the boundary of the window after the projection:

\begin{definition}
Consider a patch $\PPP$ of a planar octagonal tiling with slope $E$ and window $W$.
For $\alpha\geq 0$, we denote by $\PPP_\alpha$ the vertices of $\PPP$ which correspond to points at distance at least $\alpha$ from the boundary of $W(\PPP)$:
$$
\PPP_\alpha:=\{x\in\PPP~|~d(\pi^\prime x,\partial W(\PPP))\geq\alpha\},
$$
and respectively
$$
[\PPP]_\alpha:=\{x\in[\PPP]~|~d(\pi^\prime x,\partial W(\PPP))\geq\alpha\}.
$$
\end{definition}
We have $[\PPP]_0=[\PPP]$, $[\PPP]_\alpha = [\PPP_\alpha]$, and $[\PPP]_\alpha\subset [\PPP]_\beta$ for $\alpha>\beta$.

\begin{definition}
Consider a patch $\PPP$ of a planar octagonal tiling with slope $E$ and window $W$.
For $\alpha\geq 0$, we denote by $Ext_i^\alpha(\PPP)$ the vertices of $[\PPP]$ whose orthogonal projection on $E^\bot$ belongs to the $i$-th edge of $W(\PPP)$ and is at distance at least $\alpha$ from each line which contains another edge of $W(\PPP)$. Also, we denote  $\cup_{i=0}^8 Ext_i^\alpha([\PPP])$ as $Ext^\alpha([\PPP])$. We call such vertices \emph{extreme worms of $\SSS$}.
\end{definition}

\begin{lemma}
	For Golden-Octagonal tilings, for any real number $\alpha > 0 $, there is a seed radius $m \in N$ and growth raduis $r\in N$, such that the self-assembly algorithm provided with an $n$-pattern $\SSS$ centered at the origin, where $n\ge m$, as a seed, and an $r$-atlas, builds at step $h$ at least  $([\SSS]_\alpha\cup Ext^{\alpha}([\SSS])\cap B(0, h)$.
	\label{lem:main}
\end{lemma}

\begin{proof}
	Given an $\alpha>0$, consider an $m$-pattern $\SSS$ such that $Ext_i^\alpha([\SSS]) \cap \SSS \neq \emptyset$ for $i=1,\dots,8$, which satisfy the following constraints:

\begin{equation}
\begin{aligned}
	&\frac{\alpha}{\sin\omega_{81}} + | \partial_1 W(\SSS_\alpha) | > |\pi^\prime e_1|,
	&\frac{\alpha}{\sin\omega_{23}} + | \partial_2 W(\SSS_\alpha) | > |\pi^\prime e_2|, \\
	&\frac{\alpha}{\sin\omega_{23}} + | \partial_3 W(\SSS_\alpha) | > |\pi^\prime e_3|,
	&\frac{\alpha}{\sin\omega_{45}} + | \partial_4 W(\SSS_\alpha) | > |\pi^\prime e_4|, \\
	&\frac{\alpha}{\sin\omega_{45}} + | \partial_5 W(\SSS_\alpha) | > |\pi^\prime e_1|, 
	&\frac{\alpha}{\sin\omega_{67}} + | \partial_6 W(\SSS_\alpha) | > |\pi^\prime e_2|, \\
	&\frac{\alpha}{\sin\omega_{67}} + | \partial_7 W(\SSS_\alpha) | > |\pi^\prime e_3|, 
	&\frac{\alpha}{\sin\omega_{81}} + | \partial_8 W(\SSS_\alpha) | > |\pi^\prime e_4|,
\end{aligned}
\label{eq:constraints_alpha}
\end{equation}
where $\partial_i W(\SSS_\alpha)$ is the $i$-the edge of $W(\SSS_\alpha)$ numerated clockwise starting from the leftmost edge, and $\omega_{ij}$ is the angle between $\partial_i W(\SSS_\alpha)$ and $\partial_j W(\SSS_\alpha)$, see Figure~\ref{fig:window}. Note that for any $n$-pattern, with $n > m$, the constraints above hold as well.

\begin{figure}[htb]
	\center{\includegraphics[width=0.75\textwidth]
	{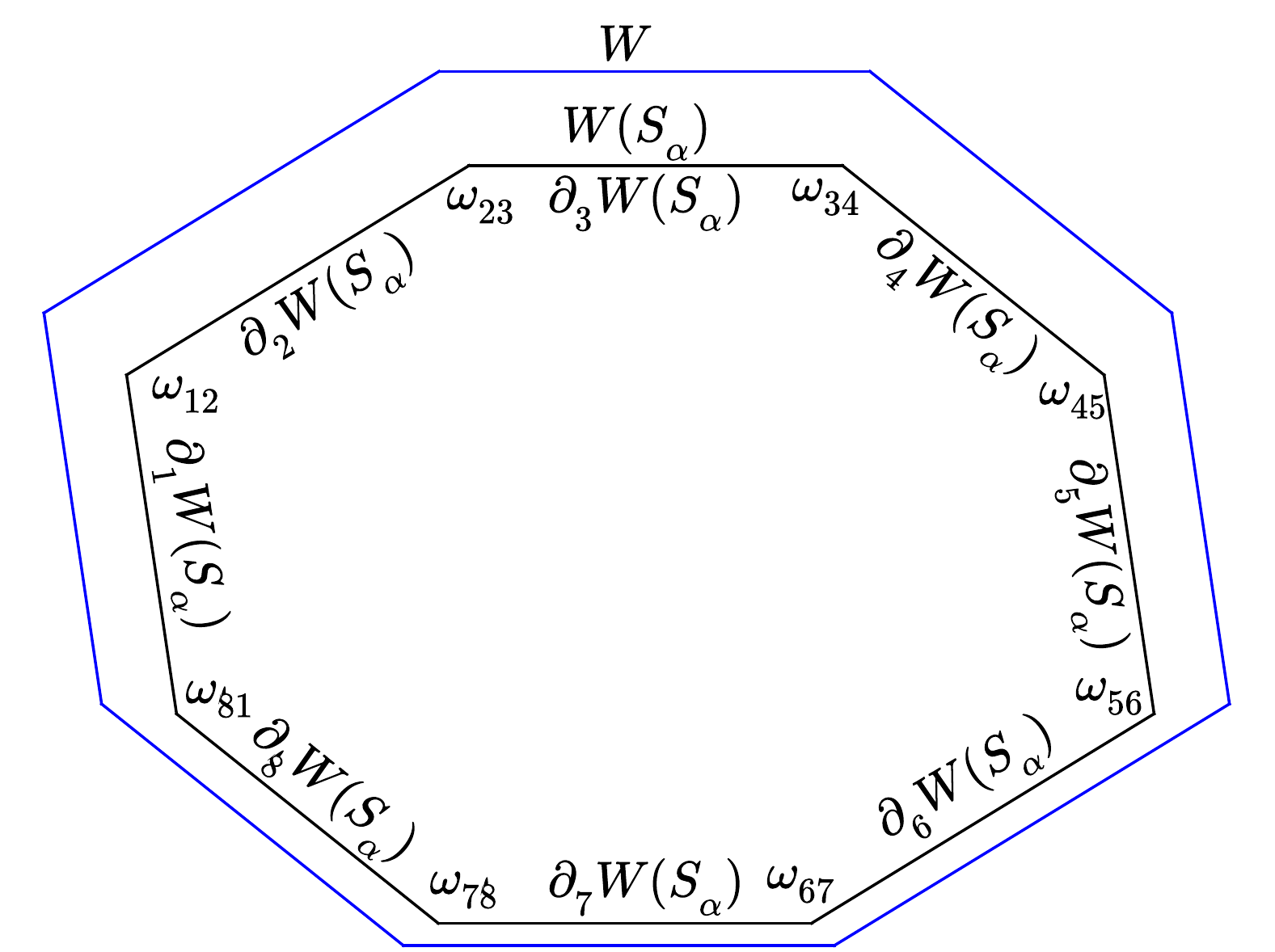}}
	\caption{\label{fig:window} Window $W$ is depicted in blue, the exact position of $W$ we do not know. Window $W(\SSS_\alpha)$ is depicted in black, with edges numerated clockwise starting from the leftmost one, $\omega_{ij}$ denotes the angle between edges $i$ and $j$.}
\end{figure}

\textbf{Base of induction.} Since $\SSS$ is an $m$-pattern, the first nontrivial step of the algorithm is with $k=nm1$. Let the growth radius $r$ be such that $B_X(r)\cap \SSS = \SSS$ for all $X\in\SSS$. Consequently, all the vertices in $\HH_{m+1}$ are forced. 

\textbf{Step of induction.} We need to show that given vertices at step $k$, all the new vertices at step $k+1$ are forced:
\[
	\HH_k := \left( [\SSS]_\alpha \cup Ext^\alpha([\SSS]) \right) \cap B_0(k) \to ( [\SSS]_\alpha \cup Ext^\alpha([\SSS]) ) \cap B_0(k+1) =: \HH_{k+1}.
\]
It is easy to see using Proposition~\ref{prop:shape} that there are no vertices in $\HH_{k+1}\setminus \HH_k$ that see less than three subperiod lines and that all the vertices which see exactly three subperiod lines are necessarily projected near the corners, see Figure~\ref{fig:shape}. We divide $\HH_{k+1}\setminus \HH_k$ into two sets: vertices which see at least $N(\alpha)$ vertices along each of its subperiod lines and, therefore, when projected on $E$ are relatively far away from corners of the growing pattern, and vertices which see $N(\alpha)$ vertices only along three of its subperiod lines, respectively whose vertices land near corners when projected on $E$.

\begin{figure}[htb]
	\center{\includegraphics[width=0.75\textwidth]
	{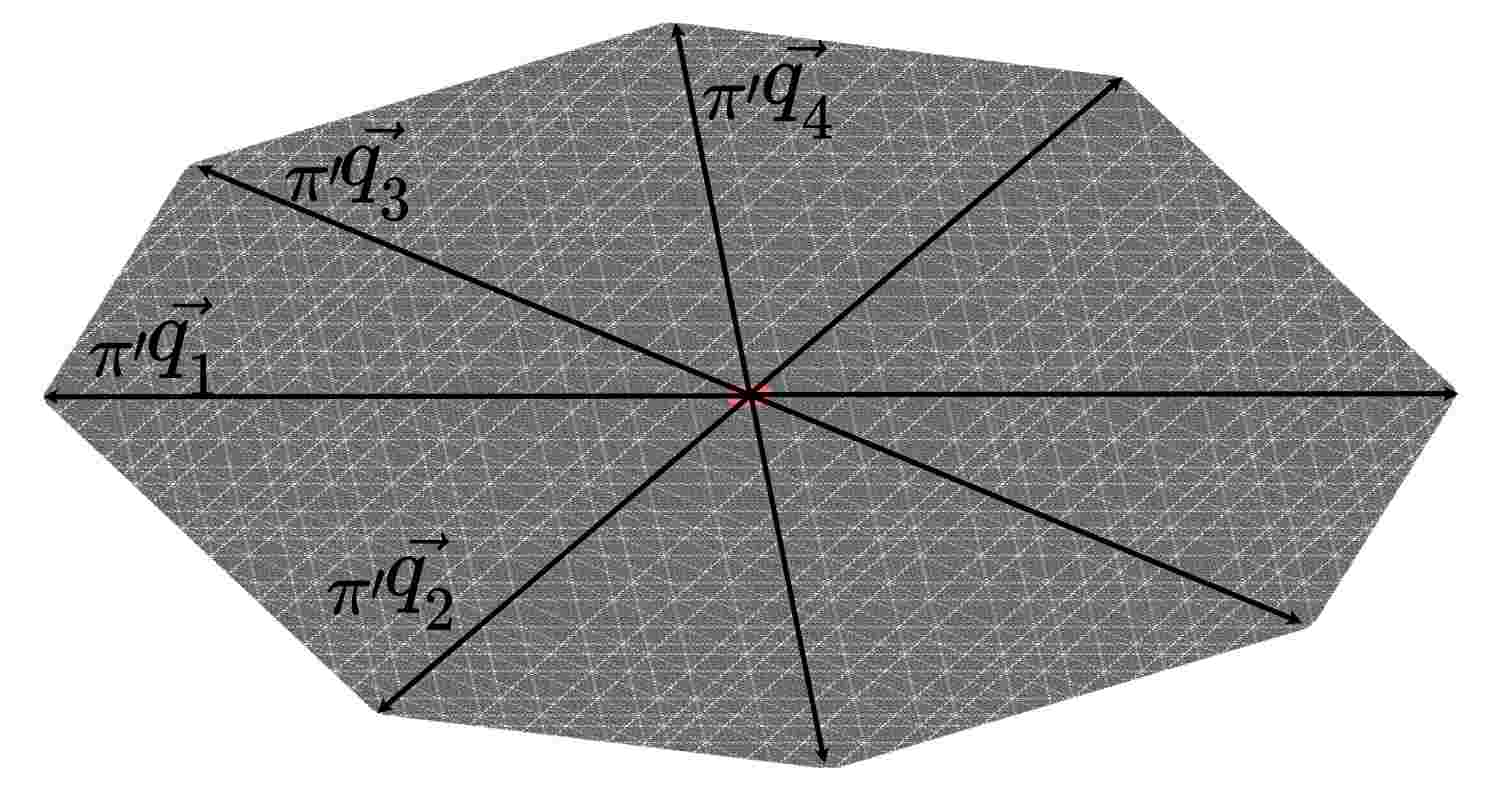}}
\caption{\label{fig:shape} The seed $\SSS$ is depicted with red. The vertices of $\HH_k$ are at a bounded distance from a convex polygon with vertices written as $\pm c_i \proj \vec{q_i}$ for $i=1,2,3,4$.}
\end{figure}

\textbf{Case 1.} It's the easier case of the two. Let $O$ be the intersection of perpendicular bisectors of $W(\SSS)$ and let $S_{ij}$ denote the sector which contains $\omega_{ij}$. Let $A \in \HH_{k+1}\setminus\HH_k$, for example, be a vertex with $\pi^\prime A \in S_{23}$, see Figure~\ref{fig:easy_case}. Consider subperiod lines $Q_2(A)$ and $Q_3(A)$. Assuming that we have chosen a big enough growth radius $r$, using Lemma~\ref{lem:folk} we state that in $r$-neighborhood of $A$ there exists two vertices $B \in Q_2(A)$ and $C \in Q_3(A)$ such that $\angle (\proj \ora{AB}, \proj  \ora{AC}) = \omega_{23}$. Since it is impossible to position window $W$ in a way that $B, C \in W$ and $A \notin W$, we conclude that $B$ and $C$ force $A$, see Figure~\ref{fig:easy_case_force}. We treat subcases when $\pi^\prime A$ projected to other sectors similarly.

\begin{figure}[htb]
	\center{\includegraphics[width=0.75\textwidth]
	{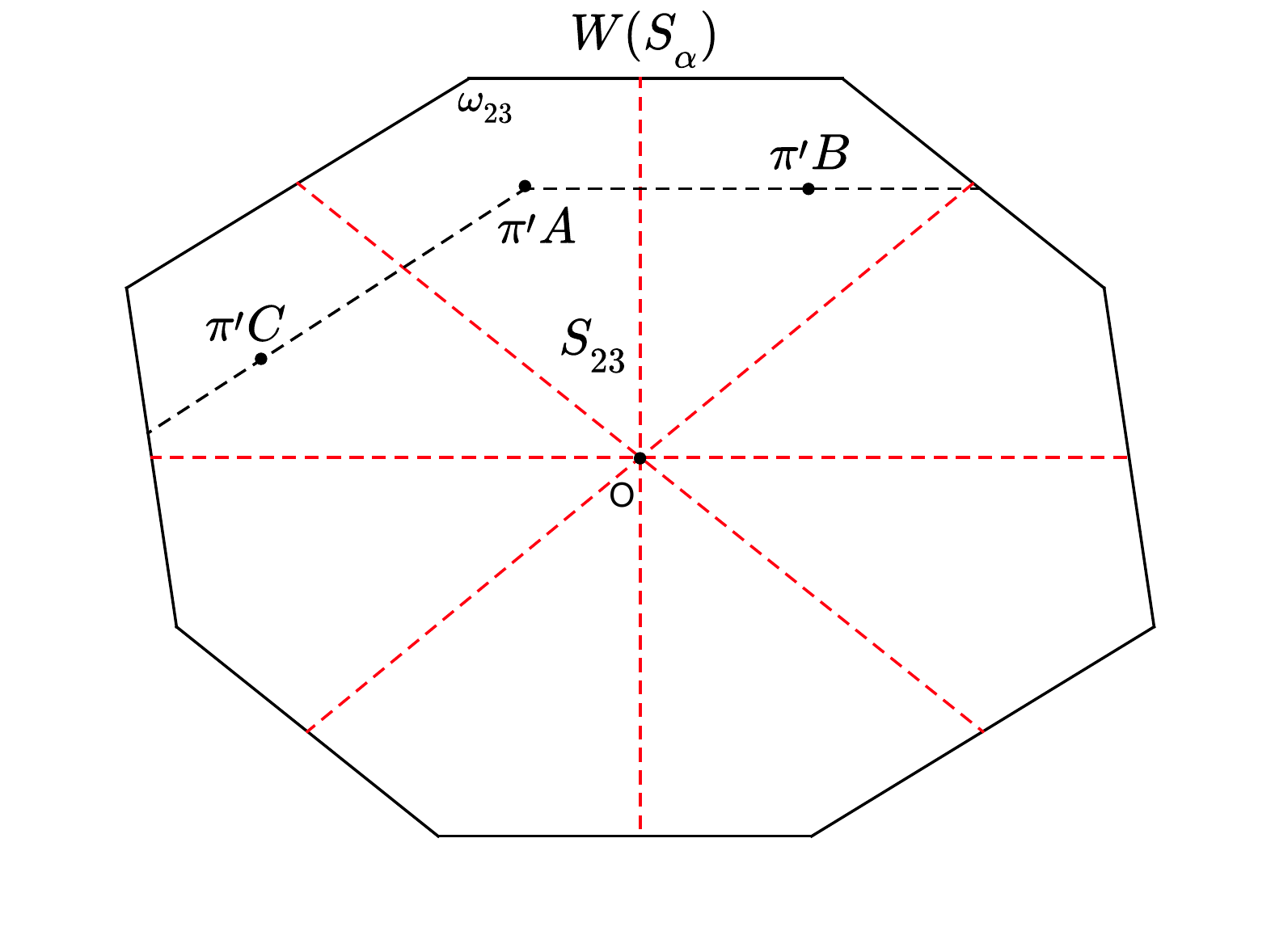}}
	\caption{\label{fig:easy_case} If $A$ belongs to the sector containing $\omega_{23}$, applying Lemma~\ref{lem:folk} we choose $r>$ such that there exists $B \in Q_3(A) \cap B_A(r) \cap \HH_k $ and $C\in Q_2(A) \cap B_A(r) \cap \HH_k$ such that $\angle (\proj \ora{AB}, \proj  \ora{AC}) = \omega_{23}$. The latter means that $B$ and $C$ force $A$.}
\end{figure}

\begin{figure}[htb]
	\center{\includegraphics[width=0.75\textwidth]
	{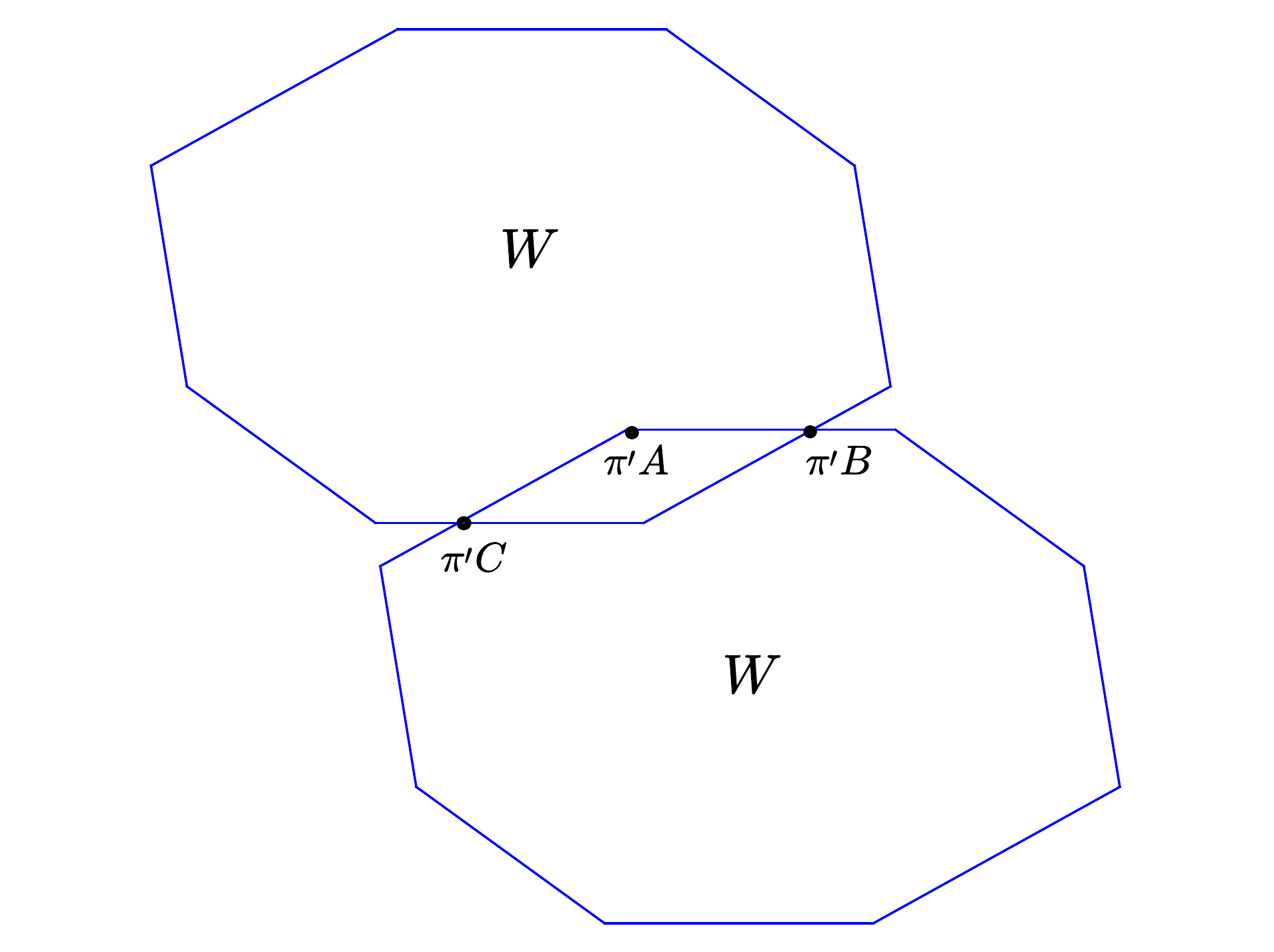}}
	\caption{\label{fig:easy_case_force} We have $B \in Q_3(A)$, $C\in Q_2(A)$ along with $\angle (\proj \ora{AB}, \proj  \ora{AC}) = \omega_{23}$. Since it is impossible to translate window $W$ in such a way that $B,C \in W$ but $A \notin W$, we conclude that vertices $B$ and $C$ force $A$. }
\end{figure}

\textbf{Case 2.} Let us say that $A \in \HH_k \setminus \HH_{k+1}$ is the vertex such that $\pi A$ is near a corner of $\HH_k$ pointed by $\pm q_2$. The latter means that $A$ sees each of its subperiod lines except for $Q_3(A)$. The proof that $A$ is forced by its neighborhood depends on the position of $\pi^\prime A$ relative to the window $W(\SSS_\alpha)$. We consider four subcases, see Figure~\ref{fig:hard_subcases}:

\begin{figure}[htb]
	\center{\includegraphics[width=0.75\textwidth]
	{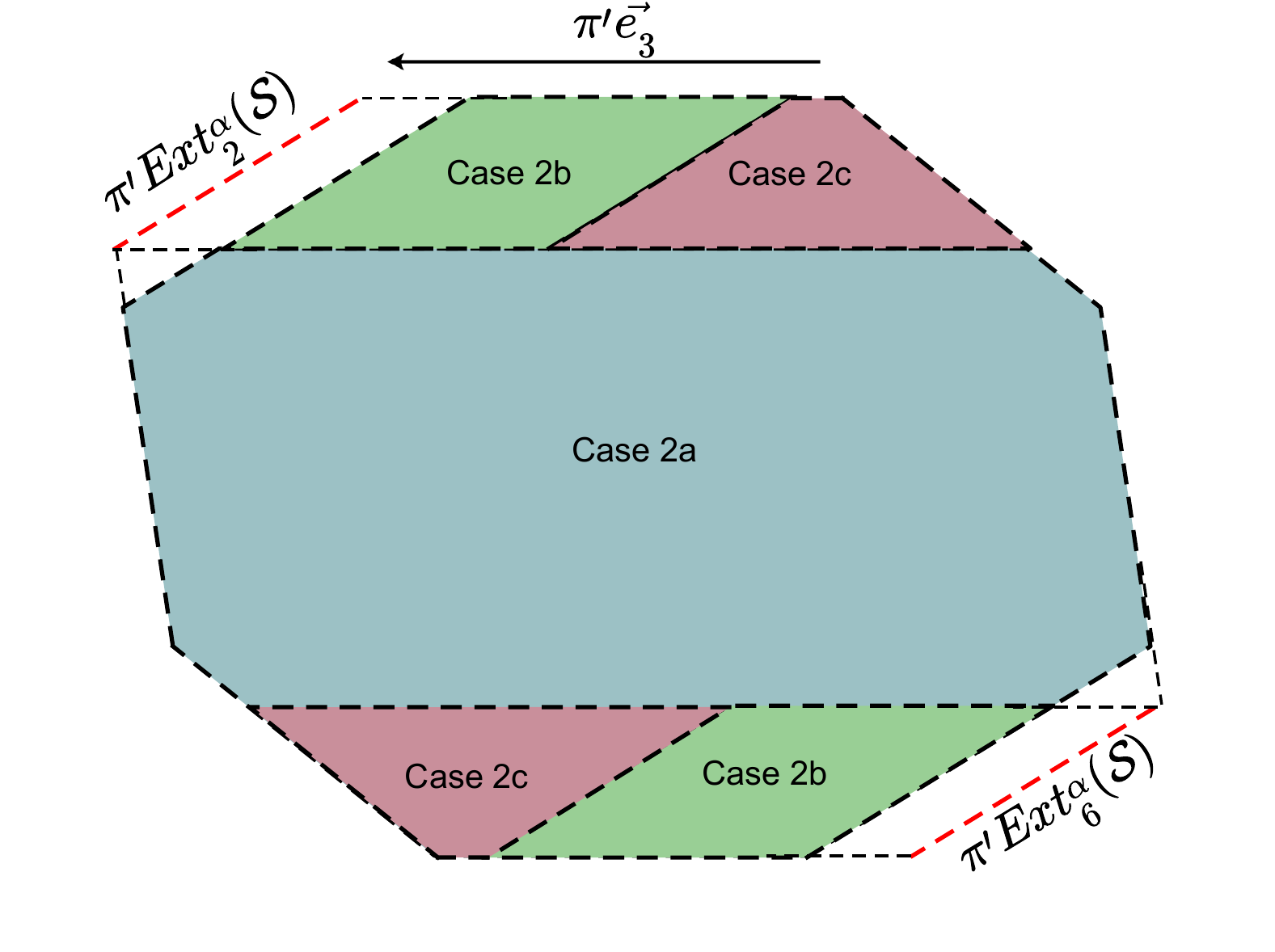}}
	\caption{\label{fig:hard_subcases} The subregions of $W(\SSS_\alpha)$ corresponding to subcases $2a-2c$ are depicted respectively with blue, green, and red. Subcase $2d$ concerning the continuation of $Ext^\alpha([\SSS])$ is depicted with red dashed line.}
\end{figure}

\textbf{Case 2a.}
Let $A$ be such that
\begin{equation}
\label{eq:hard_subcase_2a}
\pi^\prime A \in W(\mathcal{S}_\alpha) \setminus ( (\proj Ext^\alpha _2( \SSS ) \oplus \R\pi^\prime e_3) \cup (\proj Ext^\alpha _6( \SSS ) \oplus \R\pi^\prime e_3) )
\end{equation}

This subcase is very similar to Case $1$ but instead of all four subperiod lines, vertex $A$ sees just three. This downside is mitigated by the fact that distance between $\pi^\prime A$ and $\partial_3 W(\SSS_\alpha)$ or $\partial_7 W(\SSS_\alpha)$  is greater than some positive constant $\varepsilon(\alpha)$. By choosing $r = r(\varepsilon)$ big enough, by Lemma~\ref{lem:folk} we are guaranteed to find in the $r$-neighborhood of $A$ two vertices $B$ and $C$ belonging to different subperiod lines such that $\angle (\proj \overrightarrow{AB}, \proj  \overrightarrow{AC})$ equals to one of $\omega_{ij}$. Consequently, $A$ is forced by $B$ and $C$, see Figure~\ref{fig:hard_subcase_2a}.

\begin{figure}[htb]
	\center{\includegraphics[width=0.75\textwidth]
	{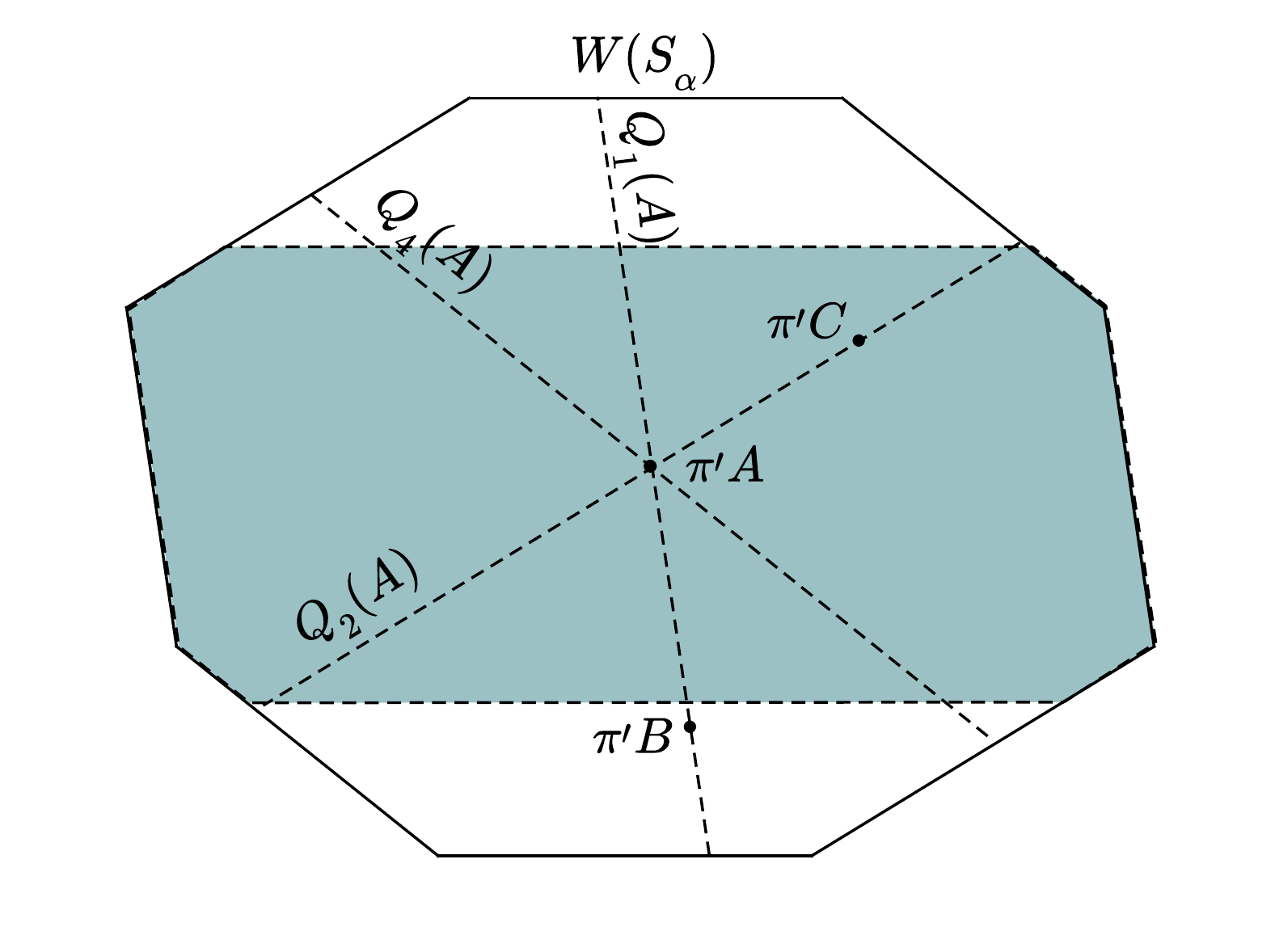}}
	\caption{\label{fig:hard_subcase_2a} Let $A$ be in subregion~\ref{eq:hard_subcase_2a} of $W(\SSS_\alpha)$. Along subperiods lines $Q_1(A),Q_2(A)$, and $Q_4(A)$, we find two vertices $B$ and $C$ such that $\angle (\proj \ora{AB}, \proj  \ora{AC}) $ equals to one of $\omega_i$. Similar to Case $1$, $B$ and $C$  force $A$.}
\end{figure}

\textbf{Case 2b.}
Let $A$ be such that
\begin{equation}
\label{eq:hard_subcase_2b}
\pi^\prime A \in W(\SSS_\alpha) \cap (  (\proj Ext^\alpha _2 (\SSS) \oplus \pi^\prime([0, -e_3])) \cup  (\proj Ext^\alpha _6 (\SSS) \oplus \pi^\prime([0, e_3]))).
\end{equation}

Since we are near the corner of the growing pattern, using Proposition~\ref{prop:shape} we state that given a large enough growth radius $r_1$, we see vertices from both extreme worms directed by $\pi q_2$ in the $r_1$-neighborhood of $A$. First, let $\pi^\prime A$ be closer to the $\partial_3 W(\SSS_\alpha)$ than to $\partial_7 W(\SSS_\alpha)$ and  let $B$ be a vertex from the $r_1$-neighborhood which belongs to the extreme worm $Ext^\alpha _2 ([\SSS])\cap\HH_k$, so we have $\pi ^\prime A$ and $\pi ^\prime B$ near abutting edges of the $W(\SSS_\alpha)$.

Consider the unit interval $(A, A+e_3)$, when projected to the perpendicular space, because of~($\ref{eq:constraints_alpha}$), it will necessarily intersect the line in the perpendicular space which contains $\pi ^\prime Ext^\alpha _2([\SSS])$, see Figure~\ref{fig:hard_subcase_2b}. Depending on the position of $B$, we choose one of two intervals $\proj (B, B+e_2), \proj (B, B-e_2)$ which intersects $\proj (A, A+e_3)$, let $X \in W$ be the intersection point. By Lemma~\ref{lem:third_intersection} we can choose $k$ so that in the $kr$-neighborhood of $A$ there is a vertex $C_1 \in [\SSS]$ such that interval $\pi^\prime (C_1, C_1 - e_4)$ also intersects $\pi^\prime (A, A + e_3)$ in the point $X$. Now consider the vertex $C_2 = C_1 - e_3$, by our initial assumptions it belongs to $[\SSS]_\alpha$ but may not belong to $\HH_k$. However, by Proposition~\ref{prop:shape}, there are vertices from $Q_4(C_2)$ in $\HH_k$, let $C \in Q_4(A)\cap \HH_k$ be the one closest to $A$. By Lemma~\ref{lem:cherez_odnu}, $B$ and $C$ forces $A$. 

If $\proj A$ is closer to the $\partial_7 W(\SSS_\alpha)$ than to $\partial_3 W(\SSS_\alpha)$, the proof remains the except for $B$ must be chosen along $Ext_6^\alpha([\SSS])$ instead of $Ext_2^\alpha([\SSS])$ and interval $(A, A-e_3)$ must be taken instead of $(A, A+e_3)$.

\begin{figure}[htb]
	\center{\includegraphics[width=0.75\textwidth]
	{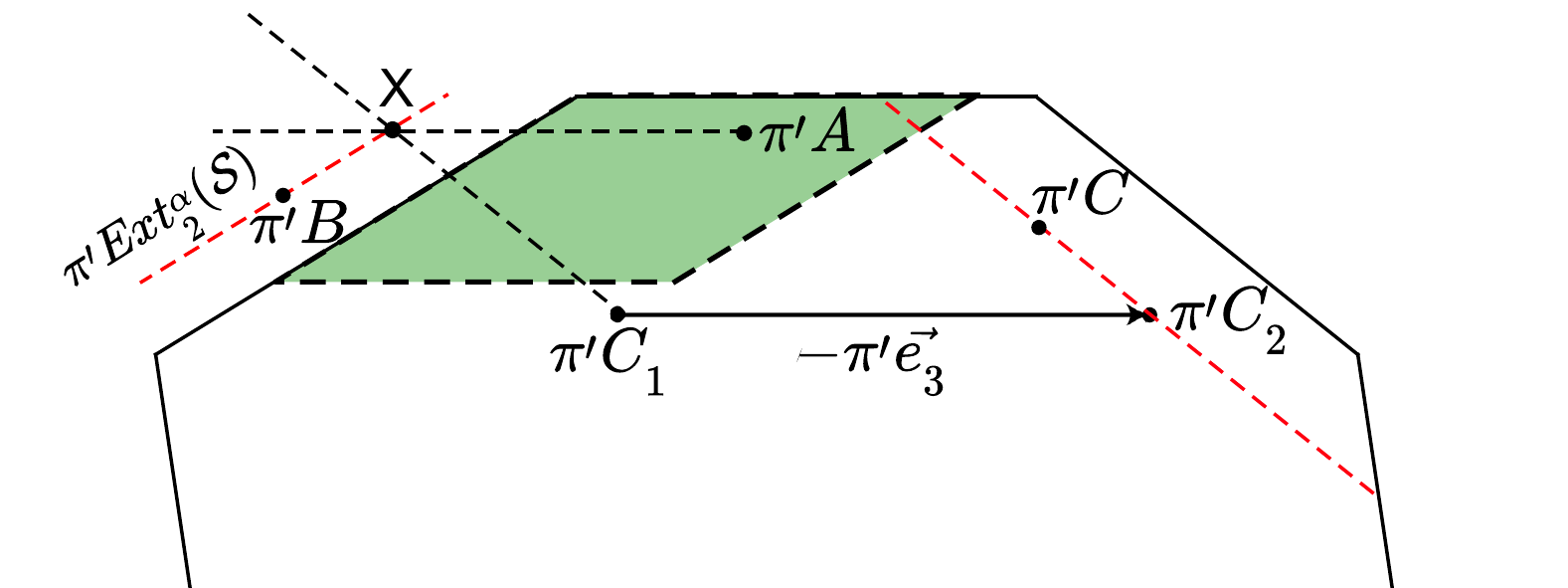}}
	\caption{\label{fig:hard_subcase_2b} In Case $2b$ when $\pi^\prime A$ projects to the subregion defined by~\ref{eq:hard_subcase_2b} depicted in green, we consider an intersection $X$ of $\proj(A, A+e_3)$ with either $\proj(B, B-e_2)$ or $\proj(B, B+e_2)$. Using Lemma~\ref{lem:third_intersection} we find interval $(C, C+e_4)$ with $\proj C_1 \in W(\SSS)$, whose projection intersects $\proj(A, A+e_3)$ also in $X$. Then, we consider $C_2 = C_1 - e_3$ and its subperiod line $Q_4(C_2)$. Let $C \in Q_4(C_2)$ be vertex from $\HH_k $ closest to $A$. Now, pair $B$ and $C$ forces $A$.}
\end{figure}

\textbf{Case 2c.}
Let $A$ be such that
\begin{align*}
	\pi^\prime A \in W(\SSS_\alpha)  \cap  
	(&(\proj Ext^\alpha _2(\SSS) \oplus \R\pi^\prime e_3) \setminus (\proj Ext^\alpha _2 (\SSS) \oplus \pi^\prime[0, -e_3]) \\  
	\cup &(\proj Ext^\alpha _6(\SSS) \oplus \R\pi^\prime e_3) \setminus (\proj Ext^\alpha _6 (\SSS) \oplus \pi^\prime[0, e_3] )),
	\numberthis
	\label{eq:hard_subcase_2c}
\end{align*}

Here we use the Lemma~\ref{lem:folk} two times consecutively. First, we choose a radius $r_1$ such that in the $r_1$-neighborhood of $A$ there is a vertex $B \in Q_3(A)$ which projects into the subregion of $W(\SSS_\alpha)$  defined by~(\ref{eq:hard_subcase_2a}), the one associated with the previous Case $2b$. Using the same reasoning as in Case $2b$, and possibly increasing the radius needed by $r_1$, we conclude that $B$ is forced by its local neighborhood.

Second, we choose $r_2$ such that in the $r_2$-neighborhood of $A$ contains a vertex $C \in \mathcal{H}_k \cap Q_4(A)$ which does not belong to the subregions~(\ref{eq:hard_subcase_2c}) and~(\ref{eq:hard_subcase_2b}) associated with Cases $2b$ and $2c$. Since $\angle (\proj \ora{AB}, \proj  \ora{AC}) = \omega_{34} = \omega_{78}$, we conclude $A$ is forced by $B$ and $C$, both of which are forced by a local neighborhood of $A$.

\begin{figure}[htb]
	\center{\includegraphics[width=0.75\textwidth]
	{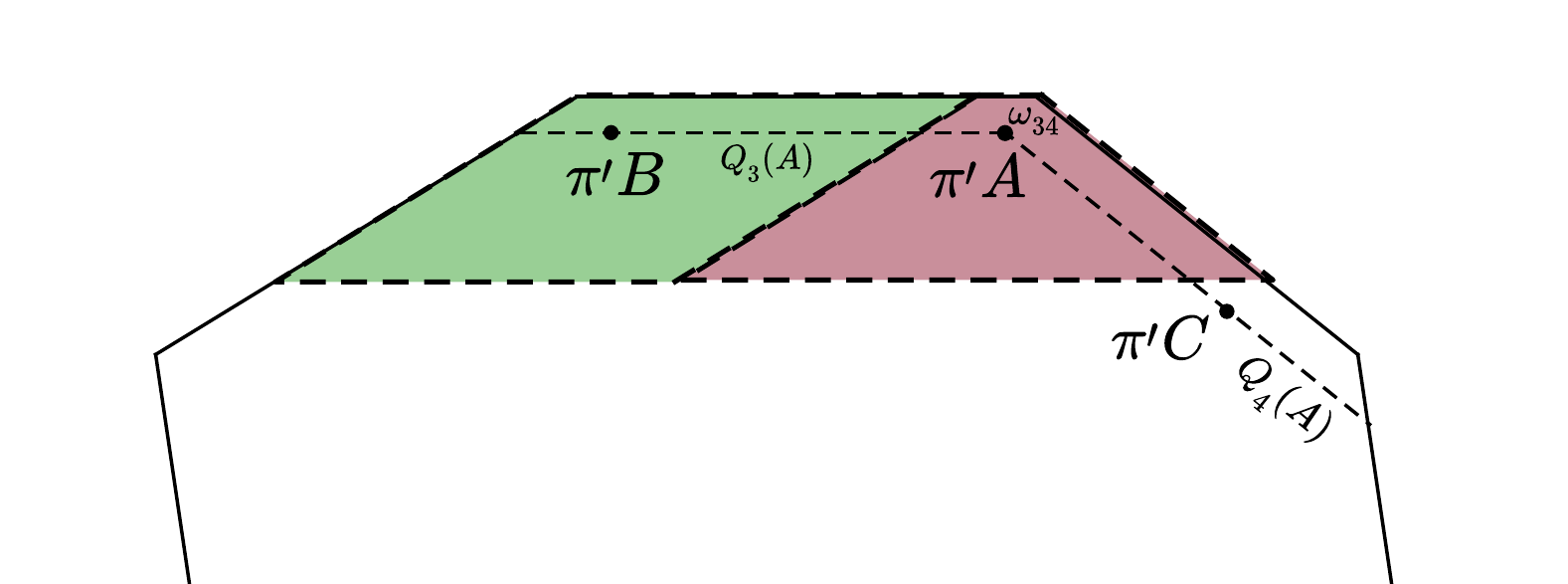}}
	\caption{\label{fig:hard_subcase_2c} Let $A$ be in subregion~\ref{eq:hard_subcase_2b} of $W(\SSS_\alpha)$. First, we choose $C \in Q_4(A) \cap \HH_k$ such that $\proj C$ does not belong to~(\ref{eq:hard_subcase_2b}) and~(\ref{eq:hard_subcase_2c}). Along subperiod line $Q_3(A)$ we find a vertex $B$ with $\proj B$ in~(\ref{eq:hard_subcase_2b}), which is forced as in Case $2b$. Since $\angle (\proj \ora{AB}, \proj  \ora{AC}) $ equals to $\omega_{34}$, $B$ and $C$ forces $A$.}
\end{figure}

\textbf{Case 2d.}

This is the last case, it is left to consider
$$
A \in (Ext_2^\alpha([\SSS]) \cup Ext_6^\alpha([\SSS])) \cap (\HH_{k+1} \setminus \HH_k)
$$
For example, let $A \in Ext^\alpha_2([\SSS]) \cap (\HH_{k+1} \setminus \HH_k) $, the proof once again depends on the position of $\pi^\prime A$ in the $W(\SSS_\alpha)$. If $\pi^\prime A$ is closer to $\partial_1 W(\SSS_\alpha)$ than to $\partial_3 W(\SSS_\alpha)$, then in the $r$-neighborhood of $A$, given that $r$ is big enough, there exists $B \in Ext_2 ^\alpha([\SSS])$ which is closer to $\partial_3 W(\SSS_\alpha)$ than $A$, as it is guaranteed by Lemma~\ref{lem:folk}. Since $A$ is near a corner of the growing pattern directed by $\vec{q_2}$, it sees every subperiod except for the third, let $B \in Q_1(A)$ be any vertex along the first subperiod which is already in place. Again, since $\angle (\proj \ora{AB}, \proj  \ora{AC}) = \omega_{12}$, $B$ and $C$ force $A$, see Figure~\ref{fig:extreme_worms_case_1}.

On the other hand, if $\pi^\prime A$ is closer to $\partial_3 W(\SSS_\alpha)$ than to $\partial_1 W(\SSS_\alpha)$, then using Lemma~\ref{lem:folk} we find a vertex $B$ which is closer to $\partial_1 W(\SSS_\alpha)$, when projected to the perpendicular space, than $A$. Then, we consider a vertex $C^\prime = A - e_3$, which belongs to $W(\SSS_\alpha)$, as it is guaranteed by~\ref{eq:constraints_alpha}, and its subperiod line $Q_4(A)$. Since $C^\prime$  also sees every subperiod except for the third, we can choose $C\in Q_4(A)\cap \HH_k$. By Lemma~\ref{lem:cherez_odnu}, $B$ and $C$ force $A$, see Figure~\ref{fig:extreme_worms_case_2}.

If $A\in Ext_6^\alpha(\SSS)$ instead of $Ext_2^\alpha(\SSS)$, the proof remains the same except that $B$ is taken along $Ext_6^\alpha([\SSS])$ and $C^\prime = A + e_3$.

That concludes the proof of Case $2$ for vertices near the two corners of the growing pattern directed by $\pi q_2$. The crucial property we used during the proof is that for vertices near a corner that do not see their $i$-th subperiod line and are projected close to $i$-th edge of $W(S_\alpha)$ (as in Case $2b$ and $2c$), to apply Lemma~\ref{lem:cherez_odnu}, the extreme worm passing through the corner must be of type $i-1 \pmod 8$ or $i+1 \pmod 8$. Fortunately for us, by Proposition~\ref{prop:shape}, this property is satisfied for every corner of the growing pattern. Consequently, the proof for vertices near other corners remains the same up to a rearrangement of indices.

Taking as seed any $n$-pattern, with $n>m$ does not change the reasoning since constraints $\ref{eq:constraints_alpha}$ remain satisfied.

\end{proof}

Theorem~\ref{thm:main} immediately follows from Lemma~\ref{lem:main}.

\begin{figure}[htb]
\centering
\begin{minipage}[t]{0.75\textwidth}
	\includegraphics[scale=0.75]{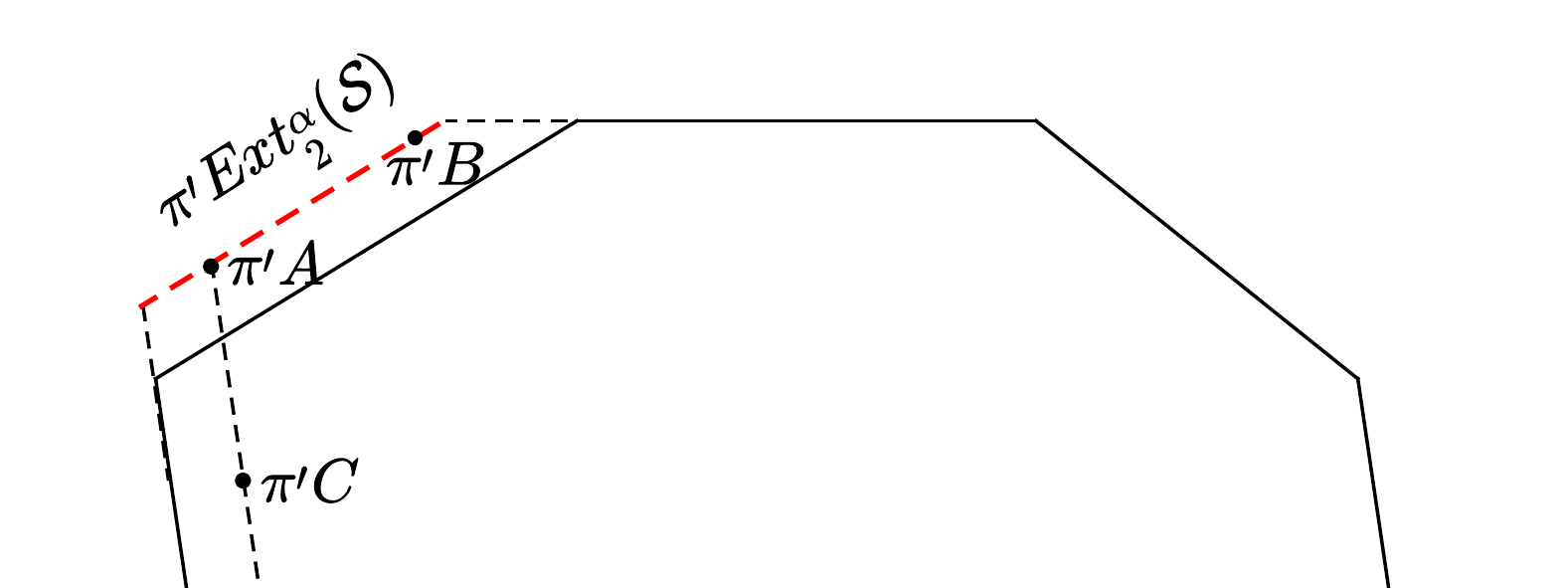}
	\caption{If $d(A, \partial_3 W(\SSS_\alpha)) > d(A, \partial_1 W(\SSS_\alpha))$, using Lemma~\ref{lem:folk} we find $B\in Ext_2 ^\alpha([\SSS])\cap \HH_k$ $B$ such that $d(B, \partial_1 W(\SSS_\alpha)) > d(A, \partial_1 W(\SSS_\alpha))$, and, since $A$ sees every subperiod except for the $3$rd, we find vertex $C \in Q_1(A)\cap \HH_k$. Pair $B$ and $C$ forces $A$.}
  \label{fig:extreme_worms_case_1}
\end{minipage}
\quad
\begin{minipage}[t]{0.75\linewidth}
	\includegraphics[scale=0.75]{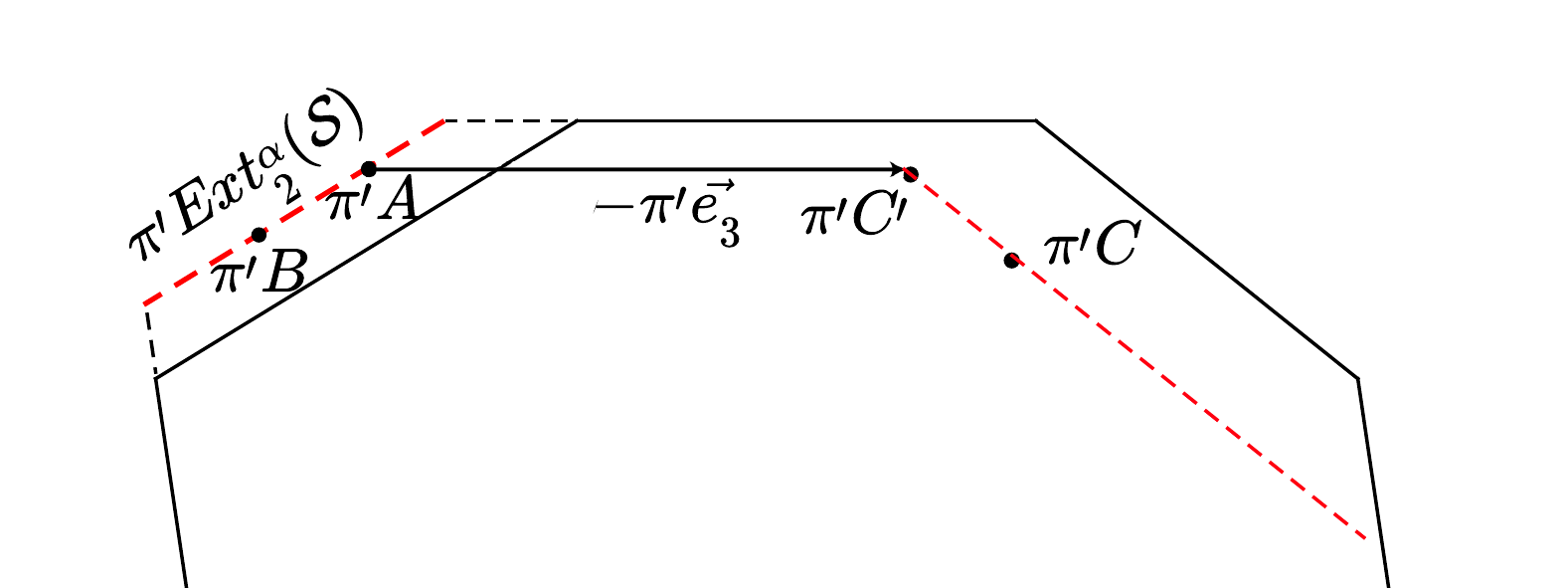}
	\caption{If $d(A, \partial_3 W(\SSS_\alpha)) < d(A, \partial_1 W(\SSS_\alpha))$, again using Lemma~\ref{lem:folk} we find $B \in Ext_2([\SSS]) ^\alpha \cap \HH_k$ such that $d(B, \partial_1 W(\SSS_\alpha)) < d(A, \partial_1 W(\SSS_\alpha))$. Then, consider a vertex $C^\prime = A - e_3 \in [\SSS_\alpha]$ and its subperiod line $Q_4(C)$. Using Lemma~\ref{prop:shape} we find vertex $C\in Q_4(C^\prime) \cap \HH_k$. Pair $B$ and $C$ forces $A$.}
  \label{fig:extreme_worms_case_2}
\end{minipage}
\end{figure}

\FloatBarrier

\begin{center}
\line(1,0){450}
\end{center}

\bibliographystyle{alpha}
\bibliography{eux2}
\end{document}